\documentclass[12pt]{article}
\usepackage{verbatim}
\usepackage{amsfonts}
\usepackage{graphics}
\usepackage{amsmath}
\usepackage{times}
\usepackage{appendix}
\usepackage{color}
\usepackage{soul}
\usepackage{mathtools}
\usepackage{enumerate}
\usepackage{fancyhdr,latexsym,amsmath,amsfonts,amssymb,amsbsy,amsthm,url}
\usepackage[margin=0.5in,footskip=0.25in]{geometry}
\usepackage{graphics,graphicx,epsfig}
\usepackage{breqn}
\usepackage{caption,subcaption,float,subfloat}
\usepackage{pdflscape}
\usepackage{hyperref}
\usepackage{tikz}
\usepackage[ruled,vlined]{algorithm2e}
\newtheorem{theorem}{Theorem}[]

\newtheorem{remark}{Remark}[]

\newtheorem{model}{Model}[]
\usepackage[english]{babel}
\makeatletter

\renewcommand{\section}{
	\@startsection
	{section}
	{1}
	{0pt}
	{1.1\baselineskip}
	{0.2\baselineskip}
	{\sc \centering}
}

\renewcommand{\subsection}{
	\@startsection
	{subsection}
	{1}
	{0pt}
	{1.1\baselineskip}
	{0.2\baselineskip}
	{\sc \centering}
}

\renewcommand{\subsubsection}{
	\@startsection
	{subsubsection}
	{1}
	{0pt}
	{1.1\baselineskip}
	{0.2\baselineskip}
	{\sc \centering}
}

\makeatother

\usepackage[flushleft]{threeparttable}
\usepackage{rotating,booktabs,multirow}
\usepackage{colortbl}
\usepackage{makecell,cellspace,caption}

\begin{document}
	
\title{\large\sc Loan portfolio management and Liquidity Risk: The impact of limited liability and haircut}
\normalsize
\author{
\sc{Deb Narayan Barik} \thanks{Department of Mathematics, Indian Institute of Technology Guwahati, Guwahati-781039, India, e-mail: d.narayan@iitg.ac.in}
\and 
\sc{Siddhartha P. Chakrabarty} \thanks{Department of Mathematics, Indian Institute of Technology Guwahati, Guwahati-781039, India, e-mail: pratim@iitg.ac.in, Phone: +91-361-2582606}}

\date{}
\maketitle
\begin{abstract}

In this article, we consider the problem of a bank's loan portfolio in the context of liquidity risk, while allowing for the limited liability protection enjoyed by the bank. Accordingly, we construct a novel loan portfolio model with limited liability, while maintaining a threshold level of haircut in the portfolio. For the constructed three-time step loan portfolio, at the initial time, the bank raises capital via debt and equity, investing the same in several classes of loans, while at the final time, the bank either meets its liabilities or becomes insolvent. At the intermediate time step, a fraction of the deposits are withdrawn, resulting in liquidation of some of the bank's assets. The liquidated portfolio is designed with the goal of minimizing the liquidation cost. Our theoretical results show that model with the haircut constraint leads to lesser liquidity risk, as compared to the scenario of no haircut constraint being imposed. Finally, we present numerical results to illustrate the theoretical results which were obtained.


\end{abstract}

\section{Introduction}

The Basel regulations, as a part of its capital framework for the international banking system, includes the Leverage Ratio as a key component of its capital requirement structure \cite{Blum2008}. The genesis of the work carried out in \cite{Blum2008} was to effectively analyze, the apparent contradiction between Leverage Ratio based capital requirement and the view of the banks that Leverage Ratio is incompatible with the modern risk management practices. The work concludes that the Leverage Ratio restrictions enable supervisors elicit honest reporting by the banks, especially when they (supervisors) have limited punitive mechanism (for errant banks) at their disposal. Hulster \cite{D2009} presents an analysis of the benefits of adopting Leverage Ratio as a part of regulatory framework, while identifying some of its inherent limitations. A formulation to compute Leverage Ratio was given by 
$\text{Leverage Ratio}=\frac{\text{Tier 1 Capital}}{\text{Adjusted Assets}}$. The benefits of using Leverage Ratio includes counter-cyclical measures, lesser regulatory arbitrage and simplicity in terms of deployment and monitoring, while undesirable incentives and being limited to balance sheet are a couple of shortcomings which were observed. An interesting observation made in \cite{Acosta2020} was that, while Leverage Ratio can act as an inducement for banks to take risk, this can be offset by the safety net of higher capital. In particular, for highly leveraged banks, the inclusion of the Leverage Ratio requirement leads to significant reduction in the levels of distress probability. A contextual description of the financial fragility resulting from being highly leveraged and the consequent inclusion and benefits of Leverage ratio is provided for in \cite{Hildebrand2008}. In \cite{Dellariccia2014}, the question of impact of interest rates on banks' leveraging tendencies is discussed. It is shown that the extent of leverage, resulting from interest rate changes is contingent on whether the capital structure is adjustable or fixed. The impact of the non-risk based Leverage capital requirement on credit practices of banks and their consequent stability is studied in \cite{Kiema2014}. The key finding reported in this work is that Leverage Ratio can potentially act as an incentive for banks to acquire highly risky credit portfolios and the authors recommend the enhancement of the current Leverage Ratio requirements.

One of the key characteristics of credit assets held by a bank, namely, it being considered as illiquid, in spite of existence of credit derivatives, which enables transfers as well as sale of loan portfolios is examined in \cite{Wagner2007}. While this practice has been welcomed by the regulatory authorities, it does not account for the possibility of greater liquidity acting as an inducement for banks to take on newer risks, and paradoxically this increases the instability and the likelihood of bank failures \cite{Wagner2007}. Another paradoxical observation made in \cite{Khan2017} finds evidence of higher risk-taking tendencies of banks having lower funding liquidity (resulting from higher deposit rates), especially during stressed market conditions. The authors of \cite{Acharya2012} carry out an evaluation of how access to abundant levels of liquidity can lead to the creation of asset bubbles, as was the case during the 2008 crisis. This scenario of banks holding on to high levels of liquidity heightens macroeconomic risk and acts as the genesis for crisis in the financial sector. Ghenimi et al. \cite{Ghenimi2017} examine the relationship between credit risk and liquidity risk and observe that (contrary to what one would expect) both these factors do not exhibit reciprocal behaviour. Further, there is no time-delayed association, even though they both contribute to the fragility of banks and their consequent instability. The global financial crisis and the consequent response of imposition of liquidity adequacy, as a part of the risk management practices was intended to ensure the stability of the banking sector \cite{Van2017}.

The safety net for banks, by way of limited liability has the potential for the banks engaging in acquiring more risky portfolios, along with concurrent undercapitalization \cite{Sinn2001}. This leads to the imposition of solvency constraints in the banking system, which, in the context of globalization of the banking system leads to a bulk of the cost of banking being borne by the domestic customer base of the banks. Limited liability, from a historical narrative, along with identification of its advantages and disadvantages are presented in \cite{Carney1998}. Limited liability offers the advantage of promoting and encouraging investments by passive investors, while extending immunity to the managers of such investments. In contrast, the concept of limited liability may encourage indulgence in risky activities, stemming from poor monitoring (by shareholders) of the risky activities being engaged in, by the firm. The concept of limited liability, particularly its impact on the banking sector, in terms of bailouts, is discussed in \cite{Cordella2003}. The paper highlights on the dilemma of the role of the central bank as the lender of last resort (LOLR), when balancing between being strongly punitive on the errant banks versus the risk of systemic impact, resulting from the failure of the banks (who enjoy immunity, by virtue of their limited liability status). The authors argue that the central bank should make funds available in situations of stressed macroeconomic conditions, as opposed to the adverse situations, resulting from imprudent portfolio decisions by the banks, and that interventions by the central banks should be contingent on these conditions.

From a mathematical perspective, optimization techniques play a vital role in the process of determination of economic factors in various paradigms. In case of Leverage Ratio \cite{Acosta2020}, the authors consider a constrained decision problem, which maximizes the expected profits (adjusted for payouts, survival and investment costs) and subject to risk-adjusted capital, as well as a Leverage Ratio constraint. In \cite{Acharya2012}, the authors consider the optimization setup for a bank owner’s problem of maximizing the expected profits, excluding any penalty resulting from liquidity crunch, subject to constraints resulting from the action of depositors. From the perspective of risk-return trade-off in case of a portfolio of loans, the author in \cite{Mencia2012} considers a risk minimization problem, akin to the classical Markowitz framework, which not only includes the expected return and sum of weights constraint, but also a probabilistic constraint, which ensures that banks achieve a return level above a threshold, chosen in a manner so as to ensure that the minimum capital requirement is met by the bank.

\section{Motivation}

In our earlier work \cite{Barik2022}, we had discussed at length, the aspects of management of a loan portfolio, in the context of incorporating limited liability. It is but natural, now to bring forth the consideration of liquidity, the relevance and importance of which can be gauged from the recent events, such as the collapse of the Silicon Valley Bank. To encapsulate the dynamics of liquidity consideration, in the existent framework, we begin with the enumeration of the necessary (and important) components towards achieving this end, that is, incorporating liquidity in the model.

The incorporation of liquidity entails the construction of ``at least'' a three (time) setup model ($t=0,1,2$), since the consideration of two time steps would not create a situation of liquidity crunch, due to the absence of the likelihood of deposit withdrawal at any intermediate time point of the loan duration (of two time points). Having considered the (aforesaid) simplest model setup, it is assumed that any liquidity requirement emerges as the result of a fraction of deposits being withdrawn at the intermediate time point $t=1$ (which for now is assumed to be known, but may be generalized to being generated through a random process, such as the exponential distribution). As a result of withdrawal demand of a fraction of the deposits, the bank is likely to face the necessity of liquidating a commensurate fraction of the loan assets (or the loan portfolio). Accordingly, in this study, the goal is to develop (and of course, analyze) a two-step approach to optimize the expected return on the loan portfolio, followed by an optimal liquidation strategy. The timeline of the events on a time scale of $t=0,1,2$ is illustrated below in Table \ref{Tab:2.0}.
\begin{table}[h]
\begin{center}
\begin{tabular}{ccc}
\hline
$t=0$ & $t=1$ & $t=2$ \\
\hline\hline
Bank collects money & A fraction of depositors &  All the risky investments \\
through debt ($d$) and &  withdraw their deposits & have matured. Hence the \\
equity ($e$) and then invests in &  and consequently the bank has  &  bank either pays its liabilities \\
safe assets and risky assets. & to liquidate some assets. & or faces insolvency/bankruptcy.\\
\hline
\end{tabular}
\caption{Timeline of events for the bank's portfolio}
\label{Tab:2.0}
\end{center}
\end{table}

\section{Model Description}

The deposit structure of the bank is assumed to follow the classical firm value assumption (due to Merton \cite{Merton74}), where the banks' assets are considered equal to the sum of the money raised through equity, $e$ (from shareholders) and the money raised through debt, $d$ (from depositors). Now, this amount is assumed to be invested in three types of loans, namely, a safe loan $L_{0}$ and two risky loans $L_{1}$ and $L_{2}$ (where $L_{1}$ is less risky than $L_{2}$). Accordingly, we have the relation:
\[e+d=\sum\limits_{i=0}^{2}L_{i}.\]

Now, coming to the quality of liquidity, one encounters the notion of high and low liquidity, corresponding to easy and hard to sell assets, respectively. It is obvious that more liquid assets have less haircut (a concept indicating the extent of loss or price impact, as a direct result of the unwinding the assets at a certain pace) and vice-versa (less liquid assets have to take more haircut). 

Accordingly, we let the constant $\gamma_{i}$ denote the haircut for loan $L_{i}$ ($i=0,1,2$). It is obvious that the more risky the asset is, the more haircut it will experience during the liquidation. Therefore,
\[\gamma_{0}<\gamma_{1}<\gamma_{2}.\]
Further, at the time $t=1$, when the liquidity requirement emerges, we assume that a certain factor $\alpha_{w}$ (of $d$) is withdrawn and another fraction $\alpha_{d}$ (of $d$) is deposited. This leads to the following observations:
\begin{enumerate}[(1)]
\item If $\alpha_{d}\ge \alpha_{w}$, then the new deposits are enough to meet the withdrawal demand of the depositors and therefore there is no need to liquidate any part of the loan portfolio.
\item If $\alpha_{d}<\alpha_{w}$ then the new deposits prove to be insufficient to meet the obligations to the depositors demanding withdrawal, which will trigger the process of liquidation from the loan portfolio. Now, the dilemma faced by the bank is as follows (summarized in Table \ref{Tab:2.1}):
\begin{enumerate}[(A)]
\item If the safe asset ($L_{0}$) is liquidated, then the likelihood of default risk at $t=2$ increases, even though this exercise means less haircut.
\item In contrast, if the riskier loans ($L_{1}$ and $L_{2}$) are liquidated then the likelihood of default risk at $t=2$ decreases, but at the cost of greater haircut.
\begin{table}[h!]
\centering
\begin{tabular}{ccc} 
\hline
Loans & Pros & Cons \\ 
\hline\hline
$L_{0}$ & Lesser haircut & More default risk \\
$L_{1}$ and/or $L_{2}$ & Less default risk & Greater haircut \\
\hline
\end{tabular}
\caption{Pros and cons for the safe asset and the risky assets}
\label{Tab:2.1}
\end{table}
\end{enumerate}
\item The model setup assumes that the equity values remain unchanged at time $t=1$.
\end{enumerate}

Suppose that we liquidate a fraction $\beta_{i}$ for loan $L_{i}$ for $0\le \beta_{i}\le 1$ ($i=0,1,2$). Then we get,
\[\sum\limits_{i=0}^{2}\beta_{i}(1-\gamma_{i})L_{i}=\left(\alpha_{w}-\alpha_{d}\right)d.\]

Before we present the description of the models which are being proposed in this work, we enumerate the various variables for the same in Table \ref{Tab:2.2}. Note that Expected Loss=Probability of Default $\times$ Loss Given Default.

\begin{table}[h!]
\centering
\begin{tabular}{cc}
\hline
Variable & Description \\
\hline\hline
$X$ & Realization of the loan portfolio \\
$\delta$ & Cost of equity \\
$\rho(x)$ & Risk of the loan portfolio $x$  \\
$L$ & The total amount of loss from liquidating the whole position \\
$\theta_{1}$ & Upper bound on risk ($> 0$) \\
$\theta_{2}$ & Lower bound on risk ($> 0$) \\
$\beta_{i}$ & Liquidation strategy for the $i$-th asset at time $t=1$ \\
$\gamma$ & Amount of loss due to liquidation \\
$k_{lev}$ & Leverage Ratio. \\
$K(x)$ & Internal Ratings Based (IRB) capital requirement for portfolio $x$ \\
$e$ & Equity component of the bank's portfolio. \\
$\eta_{i}$ & Loss Given Default (LGD) for the $i$-th loan \\
$EL_{i}$ & Expected Loss (EL) for the $i$-th loan \\
\hline
\end{tabular}
\caption{Description of the model variables}
\label{Tab:2.2}
\end{table}

\begin{model}
\label{Model2.1}	
The goal is to solve the maximization problem:
\[\max_{x,e}\left[\mathbb{E}\left[\max\left(X-(1-e),0\right)\right]-\delta e\right],\]
subject to the constraints of:
\begin{enumerate}[(A)]
\item $0 \leq x_{i} \leq 1~\forall~i=0,1,2$ (short selling is not permissible).
\item $\displaystyle{\sum\limits_{i=0}^{2} x_{i}=1}$.
\item $e \geq \max \left(k_{lev}, K(x)\right)$.
\item $\rho(x) \leq \theta_{1}$ (upper bound on risk).
\item $\displaystyle{\sum\limits_{i=0}^{2} x_{i}\gamma_{i}\leq L}$
(upper bound on the total amount of haircut if the entire portfolio is liquidated). 
\end{enumerate}
\end{model}

\begin{model}
\label{Model2.2}	
Now we construct another model like the preceding one by removing the last constraint, namely, the upper bound on to the total haircut amount, in the event of liquidation of the entire portfolio. Accordingly, the goal is to solve the maximization problem:
\[\max_{x,e}\left[\mathbb{E}\left[\max\left(X-(1-e),0\right)\right]-\delta e\right],\]
subject to the constraints of:
\begin{enumerate}[(A)]
\item $0 \leq x_{i} \leq 1~\forall~i=0,1,2$ (short selling is not permissible).
\item $\displaystyle{\sum\limits_{i=0}^{2} x_{i}=1}$.
\item $e \geq \max \left(k_{lev}, K(x)\right)$.
\item $\rho(x) \leq \theta_{1}$ (upper bound on risk).
\end{enumerate}
\end{model}

\begin{model}
\label{Model2.3}
The formulation of this model deals with designing the strategy of liquidating the problem of the loans portfolio. Let us denote the present value of the $i$-th loan at time $1$ by $X_{i}^{(1)},~i=0,1,2$. Then the goal is to solve the minimization problem:
\[\min_{\beta}\left[\sum\limits_{i=0}^{2}\beta_{i}\gamma_{i}\right],\] 
subject to the constraints of:
\begin{enumerate}[(A)]
\item $0 \leq \beta_{i} \leq x_{i},~i=0,1,2$. 
\item $\sum\limits_{i=0}^{2}\beta_{i}(1-\gamma_{i})X_{i}^{(1)}=\left(\alpha_{w}-\alpha_{d}\right)d$ (payment for the withdrawals by depositors).
\item $\displaystyle{\sum_{i=0}^{2} \beta_{i} \text{EL}_{i} \geq \theta_{2}}$ (lower bound on risk of liquidated portfolio, $\theta_{2}$).
\end{enumerate}
\end{model}
 
\begin{model}
\label{Model2.4}
Next we construct the model without the risk-lower bound. Accordingly, we get the following the minimization problem:
\[\min_{\beta}\left[\sum\limits_{i=0}^{2}\beta_{i}\gamma_{i}\right],\] 
subject to the constraints of:
\begin{enumerate}[(A)]
\item $0 \leq \beta_{i} \leq x_{i},~i=0,1,2$. 
\item $\sum\limits_{i=0}^{2}\beta_{i}(1-\gamma_{i})X_{i}^{(1)}=\left(\alpha_{w}-\alpha_{d}\right)d$ (payment for the withdrawals by depositors).
\end{enumerate}
\end{model} 
The motivation of adding the last constraint in Model \ref{Model2.3} is given as Remark \ref{Remark2.1} as follows.

\begin{remark}
\label{Remark2.1}	
Banks will (naturally) try to minimize the loss from liquidating the loan portfolio. Since the safe asset is the most liquid in our loan portfolio set-up, hence liquidating this asset leaves the riskier assets in the portfolio for the next step ($t=2$). Consequently, in case of any unfavorable conditions, the bank has to face a significant loss, if the riskier loan are not repaid in a timely manner. Keeping this in mind, we include the constraint (C) in Model \ref{Model2.3} to have a cap on liquidating the safe asset.
\end{remark}

\begin{remark}
\label{Remark2.2}
The present value of the $i$-th loan is given by:
\[X_{i}^{(1)}=\left(1+\frac{r_{i}}{2}\right)X_{i}^{(0)},\] 
and the final value of the loan is given by: \[X_{i}^{(2)}=\left(1+r_{i}\right)X_{i}^{(0)},\] 
where $r_{i}$ is the interest rate in the $i$-th loan, for $i=0,1,2$. Obviously, $r_{1}<r_{2}<r_{3}$, for $i=0,1,2$.
\end{remark}

\section{Model Analysis}

\begin{theorem}
\label{Theorem2.1}
The solution of Model \ref{Model2.1} exists.
\end{theorem}
\begin{proof}
Let $S$ be the region defined by its constraints in the model \ref{Model2.1}. Since the portfolio $x=(1,0,0)$, with equity being $100\%$ is feasible, therefore $S\neq \phi$ (because the haircut for the safe loan is $\gamma_{0}=0$, the expected loss is zero, and $100\%$ equity satisfies all the capital requirement conditions). Further, since $S$ is a closed and bounded set in $\mathbb{R}^{3}$, therefore it is a compact set. The objective function lies in ${C}(\mathbb{R}^{3})$. So Weierstrass Theorem \footnote{If $f(x)$ is continuous on a nonempty feasible set $S$, which is closed and bounded, then $f(x)$ has a global minimum in $S$.} assures that the solution to the problem exists, which we denoted by $\left(x^{*},e_{0}^{*}\right)$, where $x^{*}=\left(x_{0}^{*},x_{1}^{*},x_{2}^{*}\right)$.
\end{proof}

\begin{theorem}
\label{Theorem2.2}
The solution of Model \ref{Model2.2} exists.
\end{theorem}	
\begin{proof}
The proof for the existence for the solution of Model \ref{Model2.2} follows on the lines of Theorem \ref{Theorem2.1}.
\end{proof}

\begin{theorem}
\label{Theorem2.3}
The total amount of loss is more sensitive in case of illiquid asset.
\end{theorem}
\begin{proof}
Total amount of loss due to liquidation is given by, 
\[L=\beta \gamma^{\top}=\sum\limits_{i=0}^{2}\beta_{i}\gamma_{i},\]
where $\beta=\left(\beta_{0},\beta_{1},\beta_{2}\right)$ and $\gamma=\left(\gamma_{0},\gamma_{1},\gamma_{2}\right))$. Now we see that $\displaystyle{\frac{\partial L}{\partial \beta_{2}} \geq \frac{\partial L}{\partial \beta_{1}} \geq \frac{\partial L}{\partial \beta_{0}}}$. Therefore increasing the weights of the risky loans will increase the losses at higher rate, as compared to the increase of losses resulting from increasing the weight of the safe asset.
\end{proof}

\begin{theorem}
\label{Theorem2.4}
Sensitivity of the solution of Model \ref{Model2.1} with $L$ (upper bound on the haircut).
\end{theorem}
\begin{proof}
Recall that, the constraint (E) of Model \ref{Model2.1} is given by, 
\[\sum\limits_{i=0}^{2}x_{i}\gamma_{i}\leq L,\]
where $L$ is the maximum limit on haircut. Accordingly, we consider the equation:
\[\sum\limits_{i=0}^{2} \frac{x_{i} \gamma_{i}}{L}=1 \Rightarrow \sum\limits_{i=0}^{2} \frac{x_{i}}{L/\gamma_{i}}=1.\]
It can be observed that, as $L$ decreases, $\displaystyle{\frac{L}{\gamma_{i}}}~(\gamma_{i} > 0)$ decreases. So the upper-bound of incorporating liquid investments reduces with the decrement of $L$. If the optimal portfolio for the investment is $\textbf{x}^{*}=({x}_{0}^{*},{x}_{1}^{*},{x}_{2}^{*})$, then its corresponding haircut is $L^{*}$. Consequently, a decrease in the upper bound of the constraint by an amount greater than or equal to $(L-L^{*})$ affects the solution, and the new solution of Model \ref{Model2.1} excludes illiquid assets.
\end{proof}

\begin{theorem}
\label{Theorem2.5}
Using Limited Liability in the model reduces the incorporation of risky loans in the portfolio.
\end{theorem}
\begin{proof}
The proof has been presented in \cite{Barik2022}.
\end{proof}

\begin{theorem}
\label{Theorem2.6}
Importance of constraint $(C)$ in Model \ref{Model2.3}.
\end{theorem}
\begin{proof}
The constraint is given by, 
\[ \sum_{i=0}^{2} \beta_{i} \text{EL}_{i} \geq \theta_{2}.\]
As we have already mentioned (in the proof of Theorem \ref{Theorem2.3}), we have $\displaystyle{\frac{\partial L}{\partial \beta_{2}} \geq \frac{\partial L}{\partial \beta_{1}} \geq \frac{\partial L}{\partial \beta_{0}}}$. Therefore liquidating the most liquid asset causes the safe (risk-free) asset less haircut. Nevertheless, in the event of bankruptcy, the bank has to face an adverse scenario due to liquidation safe assets. Therefore, if we add this constraint in the model, with the risk bound $\theta_{2}$ (which is less than $\theta_{1}$, because the portfolio has risk measure less than  $\theta_{1}$ by Model \ref{Model2.1}), then a suitable threshold of this $\theta_{2}$ will restrict the liquidation of safe assets and make it incumbent on the bank to liquidate the illiquid assets (risky loans). The mathematical justification for the preceding result is as follows. We first consider the plane given by the following equation:
\[\sum\limits_{i=0}^{2}\beta_{i}EL_{i}=\theta_{2} \Rightarrow \sum\limits_{i=0}^{2} \frac{\beta_{i} EL_{i}}{\theta_{2}}=1 \Rightarrow \sum\limits_{i=0}^{2}\frac{\beta_{i}}{\left(\frac{\theta_{2}}{EL_{i}}\right)}=1.\]
We observe that the upper half of this plane contains the feasible region for Model \ref{Model2.3}. As the $\theta_{2}$ increases, $\displaystyle{\frac{\theta_{2}}{EL_{i}}}$ also increases, provided $EL_{i}>0$. Therefore the lower bound for liquidating the risky assets increases, and hence the bank has to liquidate the more risky asset for higher value of $\theta_{2}$.

From the constraint $(B)$ in Model \ref{Model2.3}, it is obvious that decreasing $\beta_{0}$ results in increase of $\beta_{1}$ and $\beta_{2}$. Differentiating the constraint $(B)$ with respect to $\beta_{0}$ we get,
\[\left(1-\gamma_{0}\right)X_{0}^{(1)}=-\frac{\partial}{\partial \beta_{0}} \left(\beta_{1}(1-\gamma_{1})X_{1}^{(1)}+\beta_{2}(1-\gamma_{2})X_{2}^{(1)}\right).\]
Let $b=\min\left(\left(1-\gamma_{1}\right)X_{1}^{(1)},\left(1-\gamma_{1}\right)X_{1}^{(1)}\right)$. Since the realizations are all positive, hence $b > 0$. Therefore, we have, 
\[\frac{\partial}{\partial \beta_{0}}\left(\beta_{1}b+\beta_{2}b\right) < 0\Rightarrow \frac{\partial}{\partial \beta_{0}}\left(\beta_{1}+\beta_{2}\right) < 0.\]
Now, we consider the worst case, in which both the risk loans have been defaulted. Let $ \beta^{(3)}=\left(\beta_{0}^{(3)},\beta_{1}^{(3)},\beta_{2}^{(3)}\right)$ be the solution of Model \ref{Model2.3} and $ \beta^{(4)}=\left(\beta_{0}^{(4)},\beta_{1}^{(4)},\beta_{2}^{(4)}\right)$ be the solution of Model \ref{Model2.4} (without this risk constraint). Therefore $\displaystyle{\beta_{0}^{(4)} \geq \beta_{0}^{(3)}}$. Now, in the worst case scenario, the realization for the solution of Model \ref{Model2.3} is better than the realization for the solution of Model \ref{Model2.4}, provided:
\[\left(\beta_{0}^{(4)}-\beta_{0}^{(3)}\right)X_{0}^{2} \geq \left(\beta_{1}^{(3)}-\beta_{1}^{(4)}\right)(1-\eta_{1})X_{1}^{0}+ \left(\beta_{2}^{(3)}-\beta_{2}^{(4)}\right)(1-\eta_{2})X_{2}^{0},\]
the LGD for the $i$-th loan is $\eta_{i},~i=0,1,2$. Finally, we can reduce $\theta_{2}$ so that the above inequality will be satisfied.
\end{proof}

\section{An Example}

In this Section, we construct an example in order to illustrate the theoretical results presented in the preceding Section. For this purpose, we consider three loans, comprising of one safe investment and two risky investments, with one of them being riskier than the other. In this context, the term riskiness refers to the credit worthiness of the debtor, which is contingent on parametric values such as PD and LGD, some of which (the parameter values) are available in \cite{Kiema2014, Shi2016}. For the purpose of our illustrative example, we have taken the values for return, PD and LGD to be the same as in \cite{Barik2022}. Further, the values of the haircuts of different types of loans are listed in the last column of the Table \ref{Tab:2.3}. It may be noted that the authors in \cite{Van2017} have discussed in detail about the haircuts applicable for various class of assets. Finally we take $\delta=1.04$, which was motivated from \cite{Kiema2014}. In summary, Table \ref{Tab:2.3} enumerates all the details of these parameter values pertaining to the loans.

We have taken a three-loan portfolio to illustrate the theoretical results obtained. One of these three loans is safe, and the others are risky. All the necessary parameter values are declared in Table \ref{Tab:2.3} with leverage ratio, $k_{lev}=4\%$.
\begin{table}[h]
\centering 
\begin{tabular}{|c|c|c|c|c|}
\hline
Loan Type & Return & PD & LGD & Haircut \\
\hline
Safe Loan & $r_{rf}=3\%$ & $0$ & $0$ & $0\%$\\
\hline
Less Risky Loan & $r_{s}=9\%$ & $p_{s}=6.1\%$ & $lgd_{s}=10\%$ & $10\%$ \\
\hline
More Risky Loan & $r_{r}=13.2\%$ & $p_{r}=12.2\%$ & $lgd_r=9\%$ & $20\%$\\
\hline
\end{tabular}
\caption{Risk parameters for the three loans}
\label{Tab:2.3}
\end{table}

We begin with the solution for Model \ref{Model2.1}, which with the inclusion of the parameter values in Table \ref{Tab:2.3}, reduces to the model:
\[\max_{x,e}\left[\mathbb{E}\left[\max\left(X-(1-e),0\right)\right]-1.04 e\right],\]
subject to the constraints of:
\begin{enumerate}[(A)]
\item $0\leq x_{i}\leq 1~\forall~i=0,1,2$.
\item $\displaystyle{\sum\limits_{i=0}^{2} x_{i}=1}$.
\item $e\geq\max\left(0.04, K(x)\right)$.
\item $x_{0}\times 0+x_{1}\times 0.0061+x_{2}\times 0.01098\left(\text{Expected Loss}\right)\leq 0.012 \left(=\theta_{1} \right) $ 
\item $\displaystyle{x_{0}\times 0+x_{1} \times 0.1+x_{2} \times 0.2\leq 0.15\left(=L\right)}$
\end{enumerate}
Further, in an analogous manner, Model \ref{Model2.2}, which does not include a cap on the haircut becomes:
\[\max_{x,e}\left[\mathbb{E}\left[\max\left(X-(1-e),0\right)\right]-1.04 e\right],\]
subject to the constraints of:
\begin{enumerate}[(A)]
\item $0\leq x_{i} \leq 1~\forall~i=0,1,2$.
\item $\displaystyle{\sum\limits_{i=0}^{2} x_{i}=1}$.
\item $e \geq \max\left(0.04, K(x)\right)$.
\item $x_{0}\times 0+x_{1} \times 0.0061+x_{2}\times 0.01098\left(\text{Expected Loss}\right) \leq 0.012 \left(=\theta_{1}\right)$. 
\end{enumerate}
 
Model \ref{Model2.1} and Model \ref{Model2.2} are continuous, non-differentiable problems (as the objective function is not differentiable). However, it can be transformed into a differentiable optimization problem, which has been elaborately discussed in \cite{Barik2022}. After solving the problem, we get the portfolio $\left(2.91\%,44.18\%,52.91\%\right)$ (that is, investments of $2.91\%$ in $L_{0}$, $44.18\%$ in $L_{1}$ and $52.91\%$ in $L_{2}$) from Model \ref{Model2.1} and for Model \ref{Model2.2} we get $\left(0\%,27.78\%,72.22\%\right)$, with leverage ratio of $4\%$, for both the portfolios. The results show that the inclusion of the cap on the haircut reduces the incorporation of illiquid assets in its portfolio. Hence Model \ref{Model2.1} outperforms \ref{Model2.2} from the perspective of liquidity risk management \footnote{We have used the ``scipy'' package of Python to solve these models}.

Next, we obtain the solution to the problem of liquidating assets, at time $t=1$, to meet the claims of the depositors, both with and without a lower bound on the risk. For this purpose, we need the solution of investment decision at $t=0$, which in turn plays the role of the first constraint. We first proceed with the solution of Model \ref{Model2.1}, that is, $\left(2.91\%,44.18\%,52.91\%\right)$. Here, in order to solve for the model, we have taken $\alpha_{w}=10\%$ and $\alpha_{d}=0$. In other words, $10\%$ of the deposits are withdrawn and there is no further investment. Hence, Model \ref{Model2.3} becomes:
\[\min_{\beta}\left[\beta_{0}\times 0+\beta_{1}\times 0.1+\beta_{2}\times0.2\right],\]
subject to the constraints:
\begin{enumerate}[(A)]
\item $0 \leq \beta_{i} \leq x_{i},~i=0,1,2,~\left(x_{0}=0.0291,x_{1}=0.4418,x_{2}=0.5291\right)$. 
\item $\beta_{0}(1-0)X_{0}^{(1)}+\beta_{1}(1-0.1)X_{1}^{(1)}+\beta_{2}(1-0.2)X_{2}^{(1)}=0.1\times0.96$.
\item $\displaystyle{\beta_{0}\times0+\beta_{1}\times 0.0061+\beta_{2}\times 0.01098\geq \theta_{2}}$
\end{enumerate}
We solve this problem with two different values of $\theta_{2}$, first for $\theta_{2}=0.05\%$ and then for $\theta_{2}=0.1\%$. Before going to the solution of Model \ref{Model2.3}, we solve Model \ref{Model2.4}, which gives the liquidation strategy without lower bound on risk. With all these parameter values, the problem becomes: 
\[\min_{\beta}\left[\beta_{0}\times0+\beta_{1}\times0.1+\beta_{2}\times0.2\right],\]
subject to the constraints:
\begin{enumerate}[(A)]
\item $0 \leq \beta_{i} \leq x_{i},~i=0,1,2,~\left(x_{0}=0.0291,x_{1}=0.4418,x_{2}=0.5291\right)$. 
\item $\beta_{0}(1-0)X_{0}^{(1)}+\beta_{1}(1-0.1)X_{1}^{(1)}+\beta_{2}(1-0.2)X_{2}^{(1)}=\left(0.1\right)\times0.96$.
\end{enumerate}
Solving Model \ref{Model2.4} we get $\left(2.91\%,7.07\%,0\%\right)$, that is, liquidation happens for $L_{0}$ and $L_{1}$. Model \ref{Model2.3}, with $\theta_{2}=0.05\%$ gives the liquidation portfolio of $\left(1.86\%,8.20\%,0\%\right)$ and for $\theta_{2}=0.1\%$ we get $\left(0\%,3.93\%,6.93\%\right)$. This supports the result established in Theorem \ref{Theorem2.6}, \textit{i.e.,} a lower bound on risk for the liquidating portfolio reduces the risk for the remaining assets, and as a result of which the chance of default is also reduced. Therefore the solution from Model \ref{Model2.3} performs better than Model \ref{Model2.4}, in the worst case scenario.

Next, we take the solution of Model \ref{Model2.2} for solving the liquidation portfolio. The formulation of the liquidation models is the same as discussed above, except for the values of $(x_{i}:i=0,1,2)$. Solving Model \ref{Model2.4} (without lower bound on the risk) gives the portfolio of $\left(0\%,10.21\%,0\%\right)$, that is, the liquidation of $10.21\%$ of less risky asset $(L_{1})$. We solve Model \ref{Model2.3}, first with a risk-lower bound of $r=0.05\%$ and then for $r=0.1\%$, resulting in the liquidation strategies of $\left(0\%,10.21\%,0\%\right)$ and $\left(0\%,3.93\%,6.93\%\right)$, respectively. Hence the model with a higher risk-lower bound reduces the chance of default at the final time, as was stated in Theorem \ref{Theorem2.6}. Now with both the inputs, it is clear that Model \ref{Model2.3} outperforms Model \ref{Model2.4} from a risk management perspective.
We have taken two inputs to solve these liquidation strategies, namely, the solution of Model \ref{Model2.1} and the solution of Model \ref{Model2.2}. From the numerical results, it is clear that, as an input, the solution of Model \ref{Model2.1} plays a better role, since it produces less probability of default at the final time, after liquidation in the intermediate time. Therefore, in summary, solving the investment decision with Model \ref{Model2.1} outperforms Model \ref{Model2.2} and in case of liquidation, Model \ref{Model2.3} is more advantageous than Model \ref{Model2.4}, which is evident from Tables \ref{Tab:2.5} and \ref{Tab:2.6}

\begin{table}[h!]
\centering
\begin{tabular}{|c|c|c|c|} 
\hline
Model & Model used as input & Portfolio with $\theta_{2}=0.05\%$ & Portfolio with $\theta_{2}=0.1\%$ \\ [0.5ex] 
\hline
Model 3 & Model 1 & $(1.86\%, 8.20\%, 0\%)$ & $(0\%, 3.93\%, 6.93\%)$ \\
Model 3 & Model 2 & $(0\%, 10.21\%, 0\%)$ & $(0\%, 3.93\%, 6.93\%)$ \\
\hline
\end{tabular}
\caption{Results for Model \ref{Model2.3}}
\label{Tab:2.5}
\end{table}

\begin{table}[h!]
\centering
\begin{tabular}{|c|c|c|} 
\hline
Model & Model which is used as input & Portfolio \\ [0.5ex] 
\hline
Model 4 & Model 1 & $(2.91\%, 7.07\%, 0\%)$  \\
Model 4 & Model 2 & $(0\%, 10.21\%, 0\%)$ \\
\hline
\end{tabular}
\caption{Results for Model \ref{Model2.4}}
\label{Tab:2.6}
\end{table}

Next, taking the solution of Model \ref{Model2.1} as the first constraint of Models \ref{Model2.3} and \ref{Model2.4}, we proceed to show the feasible region for the Problems \ref{Model2.3} and \ref{Model2.4}, respectively. We plotted three constraints of the Model \ref{Model2.3} in three different figures. The first constraint is presented in Subfigure \ref{Fig:2_FR_First}. The interior of the cube represents the set of admissible solutions. Subfigure \ref{Fig:2_FR_Second} gives the second constraint. The points on the surfaces are the feasible points, implying that the bank will liquidate its assets equivalent to the amount to be withdrawn by the depositors. Finally, Subfigure \ref{Fig:2_FR_Third} shows the third constraint. The surface presents the lower bound on risk. This model takes the upper half of this surface for an admissible liquidating strategy. Thus the intersection of these three is the feasible region for Model \ref{Model2.3}. Then we present the feasible area of Model \ref{Model2.4} in Figure \ref{Fig:2_FR_M4}. It excludes the lower bound on risk. All the other descriptions are the same as the previous scenario. The intersection of these regions is the feasible region for Model \ref{Model2.4}. The constraint (E) in Model \ref{Model2.1} is presented in Figure \ref{FIG:2_Haircut}. Here, $A_{1}$ represents an upper bound of $10\%$ in haircut, that is, the lower half of this surface is feasible whereas $A_{2}$ represents the $15\%$ upper bound of the haircut. Finally, Figures the $x$, $y$ and $z$-axes represent $L_{0}$, $L_{1}$ and $L_{2}$, respectively.  
\begin{figure}
\begin{subfigure}{.6\textwidth}
\includegraphics[width=.7\linewidth]{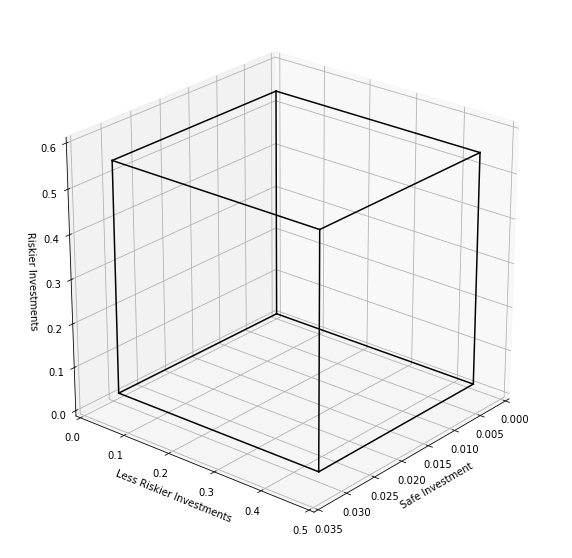}  
\caption{First Constraint}
\label{Fig:2_FR_First}
\end{subfigure}
\begin{subfigure}{.6\textwidth}
\includegraphics[width=.7\linewidth]{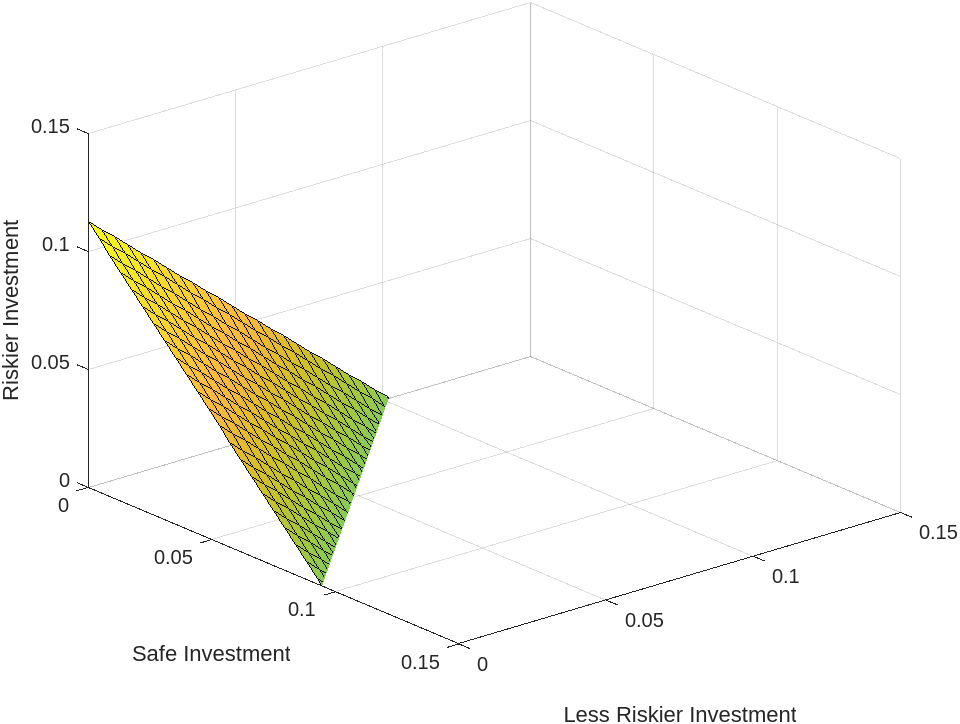}  
\caption{Second Constraint}
\label{Fig:2_FR_Second}
\end{subfigure}
\begin{subfigure}{.6\textwidth}
\hspace{3.5cm}

\includegraphics[width=.7\linewidth]{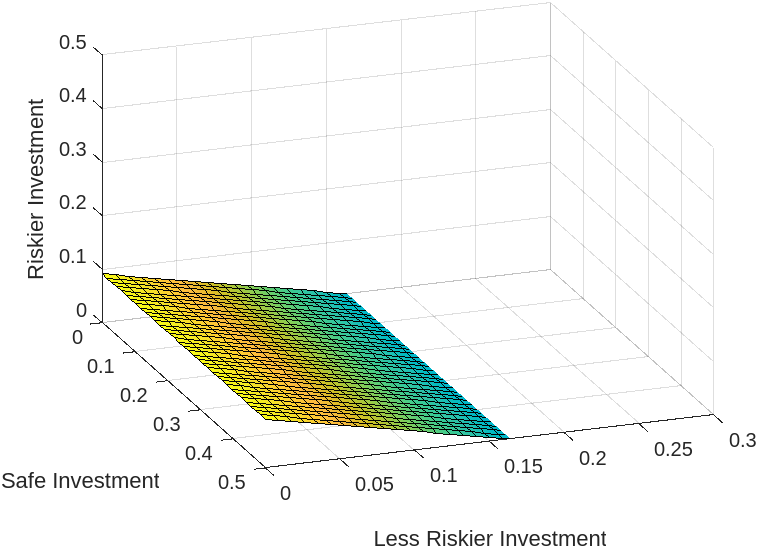}  
\caption{Third Constraint}
\label{Fig:2_FR_Third}
\end{subfigure}
\caption{The feasible region for Model \ref{Model2.3}}
\label{Fig:2_FR_M3}
\end{figure}

\begin{figure}
\begin{subfigure}{.6\textwidth}
\includegraphics[width=.8\linewidth]{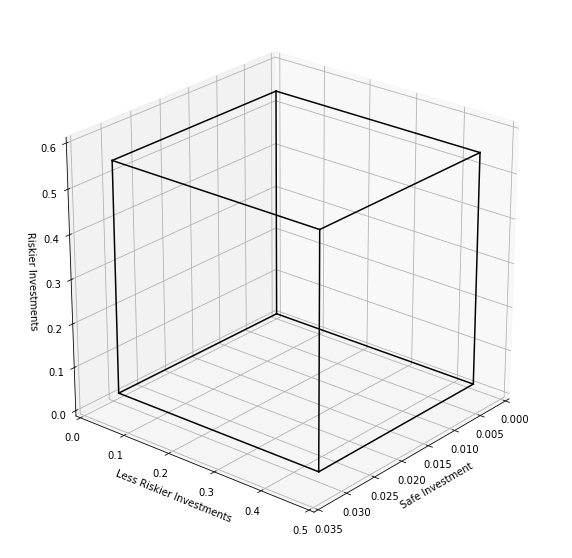}  
\caption{First Constraint}
\label{first}
\end{subfigure}
\begin{subfigure}{.6\textwidth}
\includegraphics[width=.8\linewidth]{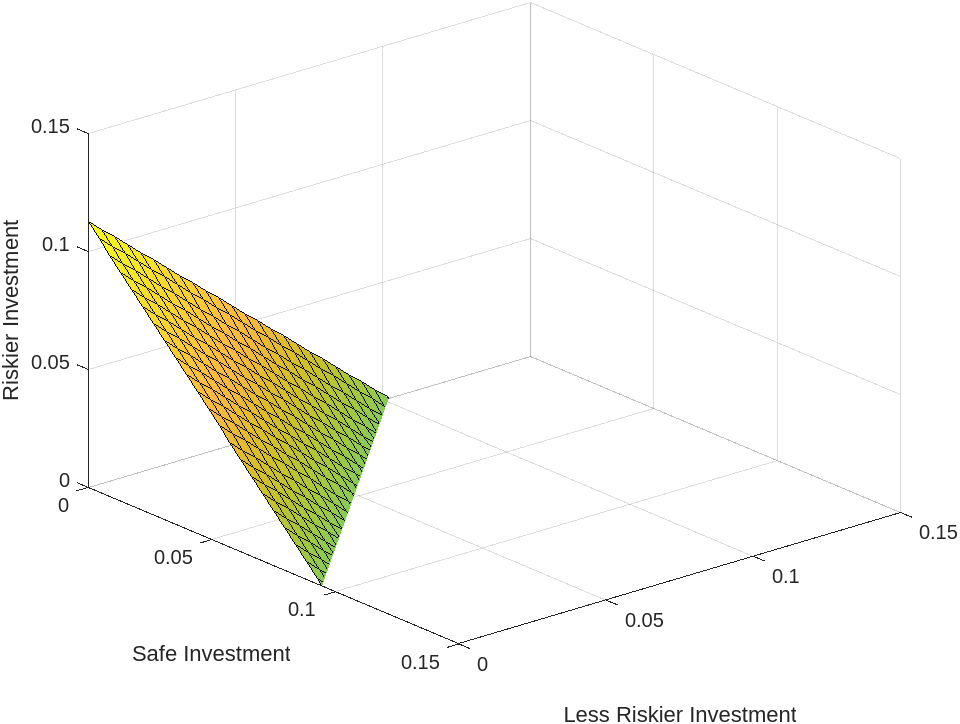}  
\caption{Second Constraint}
\label{second}
\end{subfigure}
\caption{The feasible region for Model \ref{Model2.4}}
\label{Fig:2_FR_M4}
\end{figure}

\begin{center}
\begin{figure}
\includegraphics[width=0.7\linewidth]{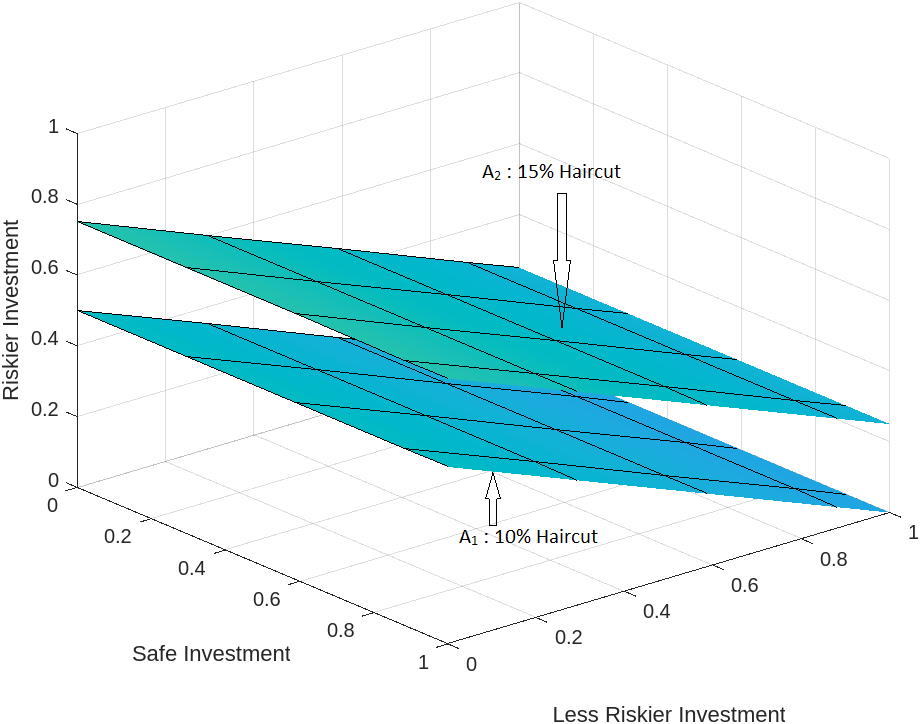}
\caption{Upper Bound on Haircut}
\label{FIG:2_Haircut}
\end{figure}
\end{center}

\section{Conclusions}

Liquidity is one of the major concerns in today's bank's portfolio decision, especially in the context of maintaining solvency and avoiding a bank-run. The inclusion of illiquid assets in the portfolio of bank, increases the likelihood of it being unable to meet its obligations to the depositors. On the other banks, by virtue of their very structure and design enjoy protections of limited liability. Accordingly, this paper presents a novel approach of incorporating the including liquidity constraint, while enjoying limited liability protection, in case of the loan portfolio held by banks. We also examine the benefits of applying a constraint on the haircut.

In this work, we have shown a three-time step model that optimizes profit with limited liability at the initial time point of $t=0$. Here, we have established the comparison between the two models, namely, Model \ref{Model2.1}) (with a cap on the haircut) and Model \ref{Model2.2} (without a cap on the haircut), with the other constraints remaining identical. At $t=1$, a fraction of depositors claim their money back. Therefore the model solves this problem of liquidating the fraction of assets. Here we compare Model \ref{Model2.3}(which includes a risk-lower bound) and Model \ref{Model2.4} (which excludes a risk-lower bound), with all other constraints being identical. Finally, at $t=2$, the bank either faces bankruptcy or makes a profit after paying all the liabilities. 

Moreover, our analysis shows that in the first comparison, the model with a cap on the total haircut includes less illiquid assets in the portfolio, as compared to the model without a cap on the total haircut. On the other hand, the inclusion of limited liability produces a less risky portfolio, which fits better into the actual scenario. In the case of the second comparison, a lower bound on the risk resists the liquidation of safe loans, as compared to the the model without a lower bound on the risk. 

Therefore, from a practitioner's point of view, the model with the cap on the haircut and limited liability can help to make the right investment decision because it reflects a more accurate scenario and produces better results. Incorporation of limited liability reduces risk and upper bound on haircut helps to survive in the stress scenario. In the intermediate time step model, the liquidation strategy minimizes the haircut. It resists the liquidation of safe assets, which increases the chance of default. Since it produces better results by satisfying all the investment criteria, therefore it is undoubtedly a more valuable model for decision-making. 

\bibliographystyle{elsarticle-num}

\bibliography{BIBLIO}

\end{document}